\documentclass[11pt]{article}
\usepackage[dvipsnames]{xcolor}
\usepackage{framed}
\definecolor{shadecolor}{rgb}{0.9,0.9,0.9}
\usepackage{mathtools}
\usepackage{amsmath}
\usepackage[shortlabels]{enumitem}
\usepackage{graphicx,epic,eepic,epsfig,amsmath,latexsym,amssymb,verbatim,color}
 
\usepackage{amsfonts}       % blackboard math symbols
\usepackage{nicefrac}       % compact symbols for 1/2, etc.

\usepackage{amsmath}
\usepackage{bbm}

\usepackage{float}
\usepackage{tikz}
\usetikzlibrary{chains}
\usetikzlibrary{fit}
\usepackage{pgflibraryarrows}		%optional
\usepackage{pgflibrarysnakes}		%optional

\usepackage{epsfig}
\usetikzlibrary{shapes.symbols,patterns} % for source symbols
\usepackage{pgfplots}

\usepackage[strict]{changepage}
\usepackage{hyperref}
\hypersetup{colorlinks=true,citecolor=blue,linkcolor=blue,filecolor=blue,urlcolor=blue,breaklinks=true}

\usepackage[marginal]{footmisc}
\usepackage{url}
\usepackage{theorem}

\theoremstyle{plain}
\newtheorem{proposition}{Proposition}
\newtheorem{lemma}[proposition]{Lemma}
\newtheorem{theorem}[proposition]{Theorem}

\theorembodyfont{\upshape}
\newtheorem{definition}{Definition}

\def\squareforqed{\hbox{\rlap{$\sqcap$}$\sqcup$}}
\def\qed{\ifmmode\squareforqed\else{\unskip\nobreak\hfil
\penalty50\hskip1em\null\nobreak\hfil\squareforqed
\parfillskip=0pt\finalhyphendemerits=0\endgraf}\fi}
\def\endenv{\ifmmode\;\else{\unskip\nobreak\hfil
\penalty50\hskip1em\null\nobreak\hfil\;
\parfillskip=0pt\finalhyphendemerits=0\endgraf}\fi}
\newenvironment{proof}{\noindent \textbf{{Proof~} }}{\hfill $\blacksquare$}

\newcounter{remark}

\newcounter{example}

\mathchardef\ordinarycolon\mathcode`\:
\mathcode`\:=\string"8000
\def\vcentcolon{\mathrel{\mathop\ordinarycolon}}
\begingroup \catcode`\:=\active
  \lowercase{\endgroup
  \let :\vcentcolon
  }

\usepackage{cleveref}
\usepackage{graphicx}
\usepackage{xcolor}

\RequirePackage[framemethod=default]{mdframed}
\newmdenv[skipabove=7pt,
skipbelow=7pt,
backgroundcolor=darkblue!15,
innerleftmargin=5pt,
innerrightmargin=5pt,
innertopmargin=5pt,
leftmargin=0cm,
rummaging=0cm,
innerbottommargin=5pt,
linewidth=1pt]{tBox}

\newmdenv[skipabove=7pt,
skipbelow=7pt,
backgroundcolor=red!15,
innerleftmargin=5pt,
innerrightmargin=5pt,
innertopmargin=5pt,
leftmargin=0cm,
rightmargin=0cm,
innerbottommargin=5pt,
linewidth=1pt]{rBox}

\newmdenv[skipabove=7pt,
skipbelow=7pt,
backgroundcolor=blue2!25,
innerleftmargin=5pt,
innerrightmargin=5pt,
innertopmargin=5pt,
leftmargin=0cm,
rightmargin=0cm,
innerbottommargin=5pt,
linewidth=1pt]{dBox}
\newmdenv[skipabove=7pt,
skipbelow=7pt,
backgroundcolor=darkkblue!15,
innerleftmargin=5pt,
innerrightmargin=5pt,
innertopmargin=5pt,
leftmargin=0cm,
rightmargin=0cm,
innerbottommargin=5pt,
linewidth=1pt]{sBox}
\definecolor{darkblue}{RGB}{0,76,156}
\definecolor{darkkblue}{RGB}{0,0,153}
\definecolor{blue2}{RGB}{102,178,255}
\definecolor{darkred}{RGB}{195,0,0}
\newcommand{\nc}{\newcommand}
\nc{\rnc}{\renewcommand}
\nc{\lbar}[1]{\overline{#1}}
\nc{\bra}[1]{\langle#1|}
\nc{\ket}[1]{|#1\rangle}
\nc{\ketbra}[2]{|#1\rangle\!\langle#2|}
\nc{\braket}[2]{\langle#1|#2\rangle}
\newcommand{\braandket}[3]{\langle #1|#2|#3\rangle}

\DeclarePairedDelimiter{\norm}{\lVert}{\rVert}
\DeclarePairedDelimiter{\abs}{\lvert}{\rvert}

\nc{\proj}[1]{| #1\rangle\!\langle #1 |}
\nc{\avg}[1]{\langle#1\rangle}
\nc{\rank}{\operatorname{Rank}}
\nc{\smfrac}[2]{\mbox{$\frac{#1}{#2}$}}
\nc{\tr}{\operatorname{Tr}}
\nc{\ox}{\otimes}
\nc{\dg}{\dagger}
\nc{\dn}{\downarrow}
\nc{\cA}{{\cal A}}
\nc{\cB}{{\cal B}}
\nc{\cC}{{\cal C}}
\nc{\cD}{{\cal D}}
\nc{\cE}{{\cal E}}
\nc{\cF}{{\cal F}}
\nc{\cG}{{\cal G}}
\nc{\cH}{{\cal H}}
\nc{\cI}{{\cal I}}
\nc{\cJ}{{\cal J}}
\nc{\cK}{{\cal K}}
\nc{\cL}{{\cal L}}
\nc{\cM}{{\cal M}}
\nc{\cN}{{\cal N}}
\nc{\cO}{{\cal O}}
\nc{\cP}{{\cal P}}
\nc{\cQ}{{\cal Q}}
\nc{\cR}{{\cal R}}
\nc{\cS}{{\cal S}}
\nc{\cT}{{\cal T}}
\nc{\cU}{{\cal U}}
\nc{\cV}{{\cal V}}
\nc{\cX}{{\cal X}}
\nc{\cY}{{\cal Y}}
\nc{\cZ}{{\cal Z}}
\nc{\cW}{{\cal W}}
\nc{\ip}{{\mathit{p}}}
\nc{\iq}{{\mathit{q}}}
\nc{\csupp}{{\operatorname{csupp}}}
\nc{\qsupp}{{\operatorname{qsupp}}}
\nc{\var}{{\operatorname{var}}}
\nc{\rar}{\rightarrow}
\nc{\lrar}{\longrightarrow}
\nc{\polylog}{{\operatorname{polylog}}}
\nc{\wt}{{\operatorname{wt}}}
\nc{\av}[1]{{\left\langle {#1} \right\rangle}}
\nc{\supp}{{\operatorname{supp}}}

\nc{\argmin}{{\operatorname{argmin}}}

\def\x{\xi}

\nc{\RR}{{{\mathbb R}}}
\nc{\CC}{{{\mathbb C}}}
\nc{\FF}{{{\mathbb F}}}
\nc{\NN}{{{\mathbb N}}}
\nc{\ZZ}{{{\mathbb Z}}}
\nc{\PP}{{{\mathbb P}}}
\nc{\QQ}{{{\mathbb Q}}}
\nc{\UU}{{{\mathbb U}}}
\nc{\EE}{{{\mathbb E}}}
\nc{\id}{{\operatorname{id}}}

\nc{\CHSH}{{\operatorname{CHSH}}}

\nc{\be}{\begin{equation}}
\nc{\ee}{{\end{equation}}}
\nc{\bea}{\begin{eqnarray}}
\nc{\eea}{\end{eqnarray}}
\nc{\<}{\langle}
\rnc{\>}{\rangle}
\nc{\rU}{\mbox{U}}

\nc{\ob}[1]{#1}

\nc{\SEP}{{\text{\rm SEP}}}
\nc{\NS}{{\text{\rm NS}}}
\nc{\LOCC}{{\text{\rm LOCC}}}
\nc{\PPT}{{\text{\rm PPT}}}
\nc{\EXT}{{\text{\rm EXT}}}
\nc{\Sym}{{\operatorname{Sym}}}

\nc{\ERLO}{{E_{\text{r,LO}}}}
\nc{\ERLOCC}{{E_{\text{r,LOCC}}}}
\nc{\ERPPT}{{E_{\text{r,PPT}}}}
\nc{\ERLOCCinfty}{{E^{\infty}_{\text{r,LOCC}}}}
\nc{\Aram}{{\operatorname{\sf A}}}

\newcommand{\eps}{\varepsilon}
\newcommand{\CPTP}{\text{\rm CPTP}}

\usepackage{tikz}
\usepackage{hyperref}
\hypersetup{colorlinks=true,citecolor=blue,linkcolor=blue,filecolor=blue,urlcolor=blue,breaklinks=true}

\makeatletter
\def\grd@save@target#1{%
  \def\grd@target{#1}}
\def\grd@save@start#1{%
  \def\grd@start{#1}}
\tikzset{
  grid with coordinates/.style={
    to path={%
      \pgfextra{%
        \edef\grd@@target{(\tikztotarget)}%
        \tikz@scan@one@point\grd@save@target\grd@@target\relax
        \edef\grd@@start{(\tikztostart)}%
        \tikz@scan@one@point\grd@save@start\grd@@start\relax
        \draw[minor help lines,magenta] (\tikztostart) grid (\tikztotarget);
        \draw[major help lines] (\tikztostart) grid (\tikztotarget);
        \grd@start
        \pgfmathsetmacro{\grd@xa}{\the\pgf@x/1cm}
        \pgfmathsetmacro{\grd@ya}{\the\pgf@y/1cm}
        \grd@target
        \pgfmathsetmacro{\grd@xb}{\the\pgf@x/1cm}
        \pgfmathsetmacro{\grd@yb}{\the\pgf@y/1cm}
        \pgfmathsetmacro{\grd@xc}{\grd@xa + \pgfkeysvalueof{/tikz/grid with coordinates/major step}}
        \pgfmathsetmacro{\grd@yc}{\grd@ya + \pgfkeysvalueof{/tikz/grid with coordinates/major step}}
        \foreach \x in {\grd@xa,\grd@xc,...,\grd@xb}
        \node[anchor=north] at (\x,\grd@ya) {\pgfmathprintnumber{\x}};
        \foreach \y in {\grd@ya,\grd@yc,...,\grd@yb}
        \node[anchor=east] at (\grd@xa,\y) {\pgfmathprintnumber{\y}};
      }
    }
  },
  minor help lines/.style={
    help lines,
    step=\pgfkeysvalueof{/tikz/grid with coordinates/minor step}
  },
  major help lines/.style={
    help lines,
    line width=\pgfkeysvalueof{/tikz/grid with coordinates/major line width},
    step=\pgfkeysvalueof{/tikz/grid with coordinates/major step}
  },
  grid with coordinates/.cd,
  minor step/.initial=.2,
  major step/.initial=1,
  major line width/.initial=2pt,
}
\makeatother

\usepackage{thmtools}
\usepackage{thm-restate}
\usepackage{etoolbox}
\makeatletter
\def\problem@s{}
\newcounter{problems@cnt}

\newcommand{\allproblems}{\problem@s}
\makeatother

\usepackage{amsmath,amsfonts,amssymb,array,graphicx,mathtools,multirow,bm,times,tcolorbox,relsize,booktabs,mathrsfs}
\usepackage[utf8]{inputenc}
\usepackage[T1]{fontenc}
\usepackage{qcircuit}
\usepackage{authblk} % package for author list
\usepackage[paperwidth=199.8mm,paperheight=297mm,centering,hmargin=20mm,vmargin=2cm]{geometry}
\usepackage[ruled, vlined, linesnumbered]{algorithm2e}
\usepackage{pifont} % for the table
\usepackage{optidef} % a package for writing optimization problems
\usepackage[caption=false]{subfig} % subfloat

\definecolor{colortwo}{rgb}{0.4,0.77,0.17}
\definecolor{colorthree}{rgb}{0.01,0.51,0.93}
\pgfplotsset{compat=1.18}

\newcommand{\set}[1]{ \left\{ #1 \right\} }

\newcommand{\op}{\operatorname}
\nc{\EPPT}{E_{\op{PPT}}}
\nc{\EPPTone}{E_{\op{PPT}}^{(1)}}
\nc{\EK}{E_{\kappa}}
\nc{\EN}{E_N}
\nc{\exact}{\op{exact}}
\nc{\PWPq}{\op{PWPq}}
\newcommand{\trace}[2][]{\tr_{#1} \left[ #2 \right]}

\newcommand{\inner}[2]{\left\langle #1, #2 \right\rangle}

\nc{\cmark}{\ding{51}}
\nc{\xmark}{\ding{55}}

\nc\bmu{{ \mathbf{u} }}
\nc\bmv{{ \mathbf{v} }}
\nc\bmp{{ \mathbf{p} }}
\nc\bmq{{ \mathbf{q} }}
\nc\bmone{{ \mathbf{1} }}

\allowdisplaybreaks

\begin{document}
\title{Physical Implementability for\\ Reversible Magic State Manipulation}

\author[1]{Yu-Ao Chen\thanks{yuaochen@hkust-gz.edu.cn}}
\author[2,3]{Gilad Gour\thanks{giladgour@technion.ac.il}}
\author[1]{Xin Wang\thanks{felixxinwang@hkust-gz.edu.cn}}
\author[1]{Lei Zhang\thanks{lzhang897@connect.hkust-gz.edu.cn}}
\author[1]{Chenghong Zhu\thanks{czhu854@connect.hkust-gz.edu.cn}}
\affil[1]{\small Thrust of Artificial Intelligence, Information Hub,\par The Hong Kong University of Science and Technology (Guangzhou), Guangdong 511453, China}
\affil[2]{\small Technion - Israel Institute of Technology, Haifa, 3200003, Israel}
\affil[3]{\small Department of Mathematics and Statistics, Institute for Quantum Science and Technology, University of Calgary, Calgary, Alberta, Canada}
\date{\today}
\maketitle

\begin{abstract}
% A fundamental problem in quantum information is understanding quantum resources operationally, with a particular emphasis on establishing a unique measure for quantifying the exact transformation rate between quantum states under practical constraints. A pivotal resource in quantum computing is the magic state, which is essential for achieving universal quantum computation. 
Magic states are essential for achieving universal quantum computation. This study introduces a reversible framework for the manipulation of magic states in odd dimensions, delineating a necessary and sufficient condition for the exact transformations between magic states under maps that preserve the trace of states and positivity of discrete Wigner representation. Utilizing the stochastic formalism, we demonstrate that magic mana emerges as the unique measure for such reversible magic state transformations. We propose the concept of physical implementability for characterizing the hardness and cost of maintaining reversibility. Our findings show that, analogous to the entanglement theory, going beyond the positivity constraint enables an exact reversible theory of magic manipulation, thereby hinting at a potential incongruity between the reversibility of quantum resources and the fundamental principles of quantum mechanics. Physical implementability for reversible manipulation provides a new perspective for understanding and quantifying quantum resources, contributing to an operational framework for understanding the cost of reversible quantum resource manipulation.
\end{abstract}

% \tableofcontents

\newpage
%%%%%%%%%%%%%%%%%%%%%%%%%%%%%%%%%%%%%%%%%%%%%%%%%%%%%%%%%%%%%%%%%%%%%%%%%
\section{Introduction}

Within the realm of quantum information science, the extent to which we can manipulate various quantum resources directly influences their practical utility. 
This challenge is often framed as the resource convertibility problem, focusing on the establishment of unique measures for transforming one resource state into another under operational restrictions.
Drawing on the principles of thermodynamics~\cite{Carnot1979}, the exploration of quantum resource reversibility has become pivotal for enhancing the manipulation of quantum resources. This subject of study encompasses a range of pioneering efforts, from initial work that defines how transformations are characterized by local operations and classical communication in the realm of pure-state entanglement~\cite{Vidal2001}, to investigations into the handling of quantum coherence under strictly incoherent operations~\cite{Winter2016}. Recent efforts also expanded these conditions across diverse contexts~\cite{gour2018quantum, Takagi_2019, regula2020benchmarking, liu2019one, wang2023reversible, regula2024reversibility}.

Due to the gap in the generalized quantum Stein's lemma~\cite{Berta2023}, the reversibility of quantum entanglement and general quantum resources under asymptotically resource non-generating operations remains unclear. 
Although the solution to the above open problem is still missing, recent advancements have shown that relaxing the trace-preserving~\cite{regula2024reversibility} or completely positive~\cite{wang2023reversible}  conditions of quantum operations can restore the reversible transformations of entanglement. Specifically, Regula and Lami~\cite{regula2024reversibility}  demonstrate that it is possible to reversibly interconvert all states in general quantum resource theories, provided that probabilistic protocols are allowed. Meanwhile, Wang et al.~\cite{wang2023reversible} show that entanglement non-generating maps enable reversible exact transformations of entangled states, leading to an exact characterization of reversible entanglement transformations governed by logarithmic negativity.
Recent works~\cite{Yuan2024,Zhao2024} have also considered generalized ways beyond quantum channels to virtually transform quantum resources, with a primary focus on shadow information or the measurement statistics of quantum resources. Within the framework of quantum resource theory, these theoretical developments facilitate the emergence of novel insights into the problem of identifying which operational properties must be compromised to attain reversibility. Furthermore, such exploration could provide new angles for quantifying and manipulating quantum resources.

The advancement of quantum information processing underscores the importance of manipulating quantum non-stabilizer resources, colloquially termed quantum magic resources, which is pivotal for fault-tolerant quantum computation (FTQC). The significance of quantum magic resources arises from the constraints of stabilizer circuits, where stabilizer circuits composed of Clifford gates subject to efficient classical simulation according to the Gottesman-Knill Theorem~\cite{gottesman1997stabilizer}. To achieve universal quantum computation and unlock potential quantum advantage~\cite{Shor_1997, knill2005quantum}, it is necessary to incorporate and manipulate magic states. 
The resource theory of magic states~\cite{Veitch2014,Howard2016,Wang2018,WWS19,Ahmadi2017,Seddon_2019,Bravyi2018,Beverland2019,Jiang2021,Leone2022} plays a crucial role in quantum computing with applications in FTQC, classical simulation of quantum circuits, and gate synthesis, yet the reversibility and asymptotic transformation rates of these states remain less explored. Understanding the extent to which magic state manipulations can be reversed and the associated costs is of fundamental importance. Investigating the reversibility of magic state manipulation will not only deepen our comprehension of these essential quantum resources but also offer novel perspectives on quantifying the magic resources in quantum systems. Moreover, studying the magic state transformation rates will shed light on the scalability and efficiency of quantum computing protocols that rely heavily on these states. 

In this paper, we introduce a reversible theory of exact magic manipulation, drawing a parallel between magic state manipulation and thermodynamics. We first show that magic can be manipulated exactly and reversibly if the allowable transformations preserve the positivity of the Wigner function in odd dimensions, referred to as PWP quasi-operations ($\PWPq$ operations). The key result is that magic mana precisely characterizes transformations of odd dimension magic states under $\PWPq$ operations (cf. Theorem~\ref{thm:reversible}), i.e.,
\begin{equation}
    \rho \xrightarrow{\PWPq} \sigma \Longleftrightarrow \cM(\rho) \geq \cM(\sigma)
.\end{equation}
This equation is particularly constructed via a new representation of magic states called the stochastic Wigner representation. We further propose the concept of the physical implementability of $\PWPq$ operations. This quantity can offer a more detailed characterization of the hardness associated with maintaining reversibility between magic states even with identical magic mana. These results together constitute the first class of transformations that guarantee the reversibility of magic state manipulation with quantitative analysis of the physical costs.

This paper is structured as follows. In Section~\ref{sec:preliminary}, we define the notation and overview the resource theory and the task of magic state manipulation. In Section~\ref{sec:sto_form}, we introduce the stochastic formalism of the Wigner representation. In Section~\ref{sec:reversibility}, we introduce a new set of operations, namely positivity Wigner preserving operations based on the stochastic formalism of Wigner representation. Then we drive a necessary and sufficient condition that fully characterize the odd dimension magic state manipulation.  Finally, we define the physical implementatbility to characterize the hardness of maintaining such reversibility in Section~\ref{sec:physical_imple}.

%%%%%%%%%%%%%%%%%%%%%%%%%%%%%%%%%%%%%%%%%%%%%%%%%%%%%%%%%%%%%%%%%%%%%%%%%
\section{Preliminary}\label{sec:preliminary}

\paragraph{Notations.} Throughout the paper, we study the Hilbert space $\cH_d$ with odd dimension $d$. Let $\cL(\cH_d)$ be the space of linear operators mapping $\cH_d$ to itself and $\cD(\cH_d)$ be the set of density operators acting on $\cH_d$. It is worth noting that qudit-based quantum computing is gaining increasing significance, as numerous problems in the field are awaiting further exploration~\cite{Wang_2020}. The case of non-prime dimension can be understood as a tensor product of Hilbert spaces each having prime dimension.

\paragraph{Properties of linear maps.} Let $\cN$ be a linear map. $\cN$ is Hermitian-preserving (HP) if it maps any Hermitian operator to another Hermitian operator; $\cN$ is positive if it maps any positive semi-definite operator to another positive semi-definite operator, and is called completely positive (CP) if this positivity is preserved on any extended reference system. $\cN$ is also trace-preserving (TP) if it preserves the trace of states.  We denote $\cJ_{\cN}$ as the Choi representation of the linear map.

\paragraph{Wigner representation.} To characterize the stabilizerness of quantum states and operations, we first recall the definition of the discrete Wigner function~\cite{WOOTTERS19871,Gross_2006a,Gross_2006b}. Let $\{\ket{j} \}_{j=0}^{d - 1}$ be the standard computational basis of $\cH_d$. The unitary boost and shift operators $X, Z \in \CC^{d \times d}$ is defined by $X\ket{j} = \ket{j + 1 \mod 2}, Z\ket{j} = w^j \ket{j}$, where $w = e^{2\pi i/d}$. The discrete phase space of a single $d$-level system is $\mathbb{Z}_d \times \mathbb{Z}_d$. At each point $\bmu=(u_1,u_2) \in \mathbb{Z}_d \times \mathbb{Z}_d$,
the discrete Wigner function of a state $\rho$ is defined as $\cW_\rho(\bmu) = \trace{A_\bmu \rho} / d$, where $A_\bmu$ is the phase-space point operator given by 
\begin{equation}
    A_\mathbf{0} = \frac{1}{d} \sum_\bmu T_\bmu, A_\bmu = T_\bmu A_0 T_\bmu^{\dagger},\, \textrm{ where }\, T_\bmu=\tau^{-u_1 u_2} Z^{u_1} X^{u_2}, \tau=e^{(d+1) \pi i / d}
.\end{equation}
Note that the construction of $A_\bmu$ is dependent on the dimension of the underlying Hilbert space. In the context of this paper, when dealing with two different Hilbert spaces labeled as $A$ and $B$ respectively, $A_\bmu$ will be denoted as $A_A^\bmu$ and $A_B^\bmu$ for each respective space.
We say a state $\rho$ has positive discrete Wigner functions (PWFs) if $\cW_{\rho}(\bmu) \geq 0, \forall\, \bmu \in \mathbb{Z}_d \times \mathbb{Z}_d$. It is known that states with non-negative Wigner functions are classically simulable and thus are useless in magic state distillation~\cite{Veitch2012}. This is analog to that of states with positive partial transpose (PPT) in entanglement distillation~\cite{Horodecki_1998, Peres1996}. The discrete Wigner function of a quantum channel is given by,
\begin{equation}
    \cW_{\cN}\left(\bmv \mid \bmu\right) = \frac{1}{d_B}\trace{ A_B^{\bmv} \cdot \cN\left(A_A^\bmu\right) }
.\end{equation} 

\paragraph{Stabilizer states and channels.}
We say the transformation is completely positive, trace-preserving and positive-Wigner-preserving (CPTP-PWP) if it is a quantum operation that completely preserves the positivity of the Wigner function of states. And it is known that transformations that preserve the positivity of Wigner function is classical simulable~\cite{WWS19}.
More details on Wigner representation can be referred to Appendix~\ref{appendix:Wigner function}. 
Some CPTP-PWP operations can preserve the set of stabilizer pure states, which are the quantum states obtained from applying Clifford operators to the zero state. Such operations are said to be completely stabilizer-preserving (CSP)~\cite{Seddon_2019}. The CSP operations further contain a special group of classical simulable operations, known as the stabilizer operations (SO), and the set of which is comprised of Clifford operations, tensoring in stabilizer states, partial trace, measurements in the computational basis, and post-processing conditioned on these measurement results.

\paragraph{Mana of quantum states.}
Within the Wigner representation for odd dimension, it is well-known that all positively represented states used in Clifford circuits admit an eﬃcient classical simulation~\cite{Mari2012}, so negativity is a necessary resource for universal fault-tolerant quantum computing~\cite{Veitch2012}.
We can quantify the amount of magic in a quantum state $\rho$ via magic mana~\cite{Veitch2014},
\begin{equation}
    \cM(\rho) = \log \sum_\bmu \abs{\cW_{\rho}(\bmu)}
.\end{equation}
Note that mana satisfies several properties such as additivity and monotonicity.

\paragraph{Magic state manipulation.}
Magic-state transformation is a crucial component in the realization of scalable, fault-tolerant, and universal quantum computation~\cite{Bravyi2005}. To gain a deeper understanding of the fundamental limits of transforming between magic states, we will investigate the exact magic transformation task. Let $\Omega$ denote a set of free operations or allowed transformations. The task of exact magic transformation describes the process of converting one state into another exactly as the number of copies approaches infinity under certain free operations. The asymptotic conversion ratio of exact magic transformation from $\rho$ to $\sigma$, under the set of free operations $\Omega$, is defined as:
\begin{equation}
    R_{\Omega} \left( \rho \to \sigma \right) = \sup \left\{r: \exists n_0: \forall n \geq n_0, \exists \Lambda_n \in  \Omega: \Lambda_n\left(\rho^{\ox n}\right)=\sigma^{\ox rn}\right\}
.\end{equation}
We say that two resource states can be exactly interconverted reversibly if
\begin{equation}
    R_{\Omega}(\rho \to \sigma) \times R_{\Omega}(\sigma \to \rho) = 1.
\end{equation}

\section{Stochastic Formalism of Magic State Transformation}\label{sec:sto_form}

In this section, we introduce the stochastic formalism of the Wigner representation and their corresponding properties. The idea of the stochastic formalism of Wigner representation has been employed to investigate the resource theory of quantum magic channels~\cite{WWS19, koukoulekidis2022constraints}. Here, we delve deeper into additional properties of the formalism.

\begin{definition}
    The \emph{stochastic Wigner representation} of an odd dimensional operator $X \in \CC^{d \times d}$ is defined as $W_X \coloneqq \frac{1}{d} \sum_{\bmu} \trace{A_{\bmu} X} \ket{u_1, u_2}$ for $\bmu = (u_1, u_2) \in \ZZ_d \times \ZZ_d$. 
    Correspondingly, the \emph{stochastic Wigner representation} of a linear map $\cN_{A \to B}$ is defined as
\begin{equation} \label{eqn:sto Wigner function}
    W_\cN \coloneqq \frac{1}{d_B} \sum_{\bmu, \bmv} \trace{ \left((A_A^\bmu)^T \ox A_B^{\bmv}\right) \cJ_\cN } \ketbra{v_1, v_2}{u_1, u_2}
.\end{equation}
\end{definition}   
The following lemma characterize the properties of the stochastic Wigner representation.
\begin{lemma} \label{lem:sto Wigner properties}
    Let $\cN_{A \to B}$ be a linear map, and $X$ be a matrix.
\begin{enumerate}
    \item $W_\cN W_X = W_{\cN(X)}$.
    \item $W_X$ is real if and only if $X$ is Hermitian.
    \item $W_\cN$ is real if and only if $\cN$ is Hermitian-preserving.
    \item $X$ is positive semidefinite if and only if $\inner{W_X}{W_{\ketbra{\psi}{\psi}}} \geq 0$ for every state $\ket{\psi}$.
\end{enumerate}
\end{lemma}
\begin{proof}
1. The result follows by the definition of stochastic Wigner representation and the properties of phase-space point operators given in Appendix~\ref{appendix:Wigner function}:
\begin{align}
    W_\cN\, W_X 
    &= \sum_{\bmu, \bmv} \cW_{\cN}\left(\bmv \mid \bmu\right) \cW_{X}\left(\bmu\right) \ket{v_1, v_2} \\
    &= \frac{1}{d_A d_B} \sum_{\bmu, \bmv} \trace{ \left(\left(A_A^\bmu\right)^T \ox A_B^\bmv \right) \cJ_\cN} \trace{A_A^\bmu X} \ket{v_1, v_2} \\
    &= \frac{1}{d_A d_B} \sum_\bmv \left( \sum_\bmu \trace{ A_B^\bmv \cdot \cN \left( A_A^\bmu \right) } \trace{A_A^\bmu X} \right) \ket{v_1, v_2} \\
    &= \frac{1}{d_B} \sum_\bmv \trace{  A_B^\bmv \cdot \cN \left( \frac{1}{d_A} \sum_\bmu \trace{A_A^\bmu X} A_A^\bmu \right) } \ket{v_1, v_2} \\
    &= \frac{1}{d_B} \sum_\bmv \trace{  A_B^\bmv \cdot \cN \left( X \right) } \ket{v_1, v_2}  
    = W_{\cN \left( X \right)}
.\end{align}

2. The statement trivially follows by the fact that phase-space point operators are Hermitian.

3. It suffices to prove the backward direction. Suppose $W_\cN$ is real. Decompose the Choi representation of $\cN$ as
\begin{equation}
\cJ_\cN = H + S,\quad \textrm{where } H = \frac{\cJ_\cN + \cJ_\cN^\dag}{2},\, S = \frac{\cJ_\cN - \cJ_\cN^\dag}{2}
.\end{equation}
Observe that $H$ and $iS$ are Hermitian. Then denote point $\bmu'$ such that $A_B^{\bmu'} = \left(A_B^\bmu\right)^T$. We have
\begin{equation}
\trace{ \left(\left(A_A^\bmu\right)^T \ox A_B^\bmv \right) \cJ_\cN} = \trace{A_{\bmu' \oplus \bmv} H} - i \trace{A_{\bmu' \oplus \bmv} \left(iS\right)}
.\end{equation}
Since $A_{\bmu' \oplus \bmv}$ is Hermitian, $\trace{A_{\bmu' \oplus \bmv} H}, \trace{A_{\bmu' \ox \bmv} \left(iS\right)} \in \RR$. Since $W_\cN$ is real,
\begin{equation}
\trace{A_{\bmu' \oplus \bmv} \left(iS\right)} = 0 \implies iS = 0 \implies \cJ_\cN = \cJ_\cN^\dag
.\end{equation}
That is, $\cJ_\cN$ is Hermitian and hence $\cN$ is HP.

4. For every matrix $Y$,
\begin{equation}
    \trace{X Y} = \trace{\sum_{\bmu, \bmv} \cW_X(\bmu) \cW_Y(\bmv) A_\bmu A_\bmv} = \sum_\bmu \cW_X(\bmu) \cW_Y(\bmu) = W_X^\dagger W_Y
.\end{equation}
Then the statement follows by the definition of positive semi-defnite,
\begin{equation} \label{eqn:state psd}
    X \succeq 0 
    \Longleftrightarrow \braandket{\psi}{X}{\psi} = \trace{X \ketbra{\psi}{\psi}} \geq 0 \,\, \forall\, \ket{\psi}
    \Longleftrightarrow \inner{W_X}{W_{\ketbra{\psi}{\psi}}} \geq 0 \,\, \forall\, \ket{\psi}
.\end{equation}
\end{proof}

%%%%%%%%%%%%%%%%%%%%%%%%%%%%%%%%%%%%%%%%%%%%%%%%%%%%%%%%%%%%%%%%%%%%%%%%%%%%%%%%%%%%%%%%%%%%%
\section{Reversibility of Exact Magic Manipulation}\label{sec:reversibility}

In this section, we venture beyond the confines of the quantum realm to thoroughly characterize the fundamental limits of magic transformation.
We establish that the monotonicity of magic mana aligns with the existence of a set of non-physical operations, called the PWP quasi-operations ($\PWPq$ operations). As shown in Figure~\ref{fig:hierarchy}, these operation are defined as the Hermitian-preserving and trace-preserving maps that completely preserve the positivity of their Wigner functions. The underlying intuition is that $\PWPq$ operations are tightly related with the stochastic transformation of quasi-probability distributions, as shown in the following statement.

\begin{proposition}\label{prop:stochastic formalism}
    $\cN_{A \to B}$ is a $\PWPq$ operation if and only if $W_\cN$ is a column stochastic matrix.
\end{proposition}
\begin{proof}
    Lemma~\ref{lem:sto Wigner properties} states that $W_\cN$ is real if and only if $\cN$ is Hermitian-preserving; Equation~\eqref{eqn:sto Wigner function} implies that $W_\cN$ is non-negative if and only if $\cN$ is positive-Wigner-preserving; one also observes that $\bmone^T W_\cN = \bmone^T$ for all $\bmu$ if and only if $\cN$ is trace-preserving, since
\begin{align}
    \trace{\cN \left( A_A^\bmu \right)} 
    &= \trace{ \left(\left(A_A^\bmu\right)^T \ox I_B \right) \cJ_\cN}
    = \sum_\bmv \frac{1}{d_B} \trace{ \left(\left(A_A^\bmu\right)^T \ox A_B^\bmv \right) \cJ_\cN} \\
    &= \sum_\bmv  \braandket{ v_1, v_2 }{W_\cN}{ u_1, u_2 }
    = \bmone^T W_\cN \ket{u_1, u_2}
\end{align}
    and $\set{ A_A^\bmu }_\bmu$ forms a basis of $\cL \left( \cH_A \right)$. We therefore conclude that $W_\cN$ is real, non-negative and its columns sum to $1$ if and only if $\cN$ is HP, PWP and TP. In other words, $W_\cN$ is a column stochastic matrix if and only if $\cN$ is $\PWPq$.
\end{proof}

In the context of the Wigner representation, quantum states can be depicted as quasi-probability distributions, which allows for the possibility of stochastic manipulation through $\PWPq$ operations. Leveraging this property, we demonstrate that a quantum state $\rho$ can be transformed to $\sigma$ under $\PWPq$ operations if and only if the mana of $\rho$ exceeds the magic mana of the target state $\sigma$. This finding establishes a parallel in characterizing magic manipulation to the role of logarithmic negativity in characterizing entanglement manipulation within entanglement theory~\cite{wang2023reversible}. The necessary and sufficient condition for this transformation are detailed in the following theorem~\ref{thm:magic transfer}. 

\begin{figure}[t]
    \centering
    \includegraphics[width=0.4\linewidth]{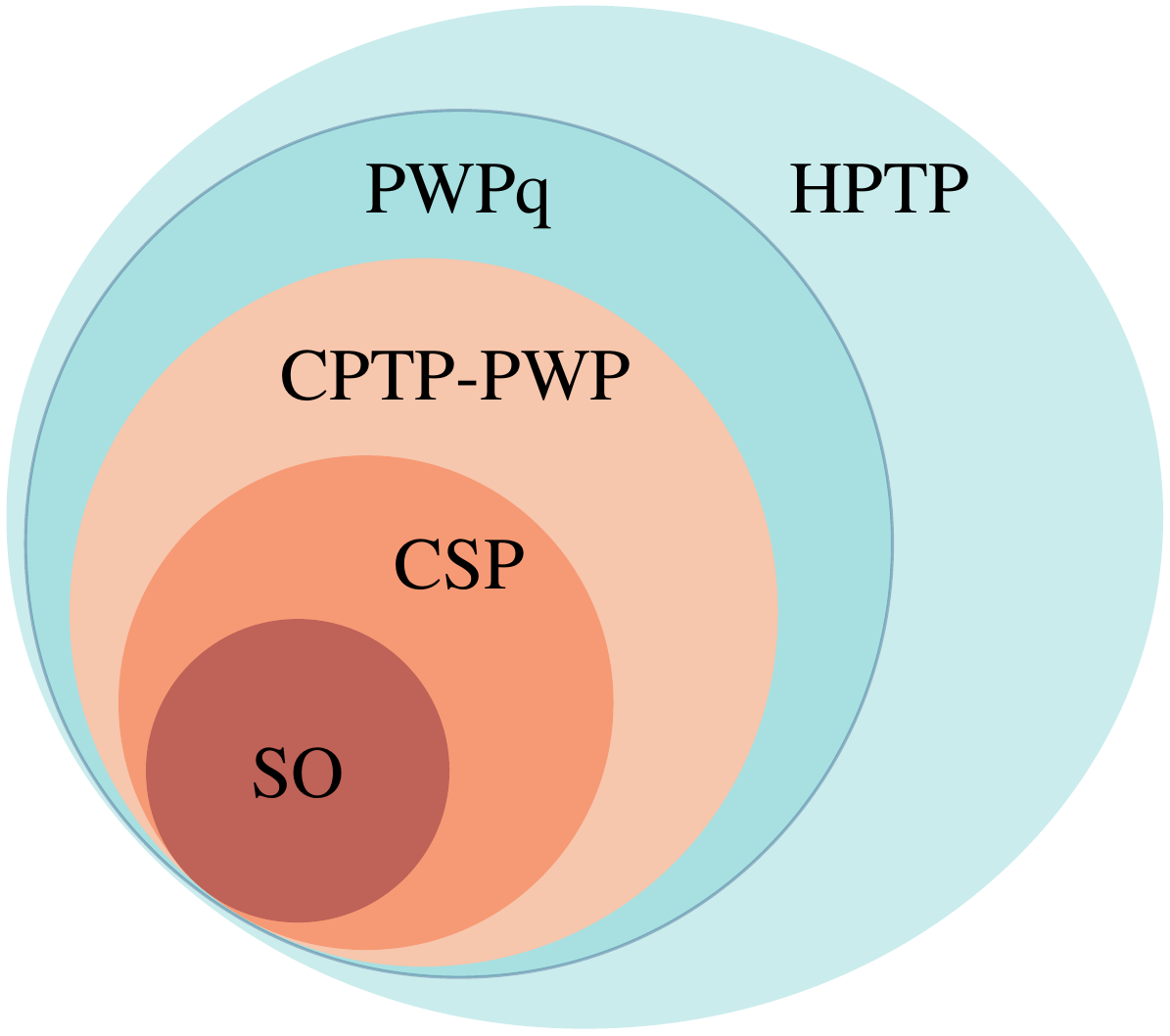}
    \begin{tabular}{cc}
    \toprule
    \midrule
    Class of Operations & Abbreviation \\
    \midrule
    Stabilizer operations & SO \\
    Completely stabilizer-preserving operations & CSP \\
    Completely positive, trace-preserving, completely positive-Wigner-preserving maps & CPTP-PWP \\
    Hermitian-preserving, trace-preserving, positive-Wigner-preserving maps  & PWPq  \\
    Hermitian-preserving and trace-preserving maps & HPTP \\
    \addlinespace
    \midrule
    \bottomrule
    \end{tabular}
    \caption{Schematic hierarchy of operations. The pictured inclusions among $\text{SO}$, $\text{CSP}$, $\text{CPTP-PWP}$, $\text{PWPq}$, $\text{HPTP}$. The main focus of this paper is exact magic manipulation under PWPq operations.}
    \label{fig:hierarchy}
\end{figure}

\begin{tcolorbox}[width = 1.0\linewidth]
\begin{theorem}\label{thm:magic transfer}
    For two odd dimensional quantum states, there exists $\PWPq$ operation $\cN$ such that $\cN(\rho) = \sigma$ if and only if
\begin{equation}
    \cM(\rho) \geq \cM(\sigma),
    \end{equation}
    where $\cM(\cdot)$ is the mana of the state.
\end{theorem}
\end{tcolorbox}
\begin{proof}
    Lemma~\ref{lem:sto Wigner properties} and Proposition~\ref{prop:stochastic formalism} states that $\cN(\rho) = \sigma$ for a HPWP-PWP operation $\cN$ if and only if $W_\cN W_\rho = W_\sigma$ for a column stochastic matrix $W_\cN$. 
    For simplicity, denote
\begin{equation}~\label{eqn:sto notation}
    W = W_\cN = \begin{bmatrix} \,w_1^T\, \\ \vdots \\ \,w_n^T\, \end{bmatrix} \in \RR^{d \times {d'}}, 
    \bmp = W_\rho = \begin{bmatrix} \,p_1\, \\ \vdots \\ \,p_m\, \end{bmatrix} \in \RR^{d'},
    \textrm{ and } \bmq = W_\sigma \in \RR^d
.\end{equation}
    Then one can translate the existence of such $\cN$ mapping $\rho$ to $\sigma$ as the existence of $\mathbf{w}\in\RR^{dd'}$ satisfying
\begin{equation}~\label{eqn:farkas constraint}
    \begin{bmatrix}
        p_1 I & p_2 I & \cdots & p_m I \\
        \bmone^T & {0} & {\cdots} & {0} \\
        {0} & \bmone^T & {\cdots} & {0} \\
        {\vdots} & {\vdots} & {\ddots} & {\vdots} \\
        {0} & {0} & {\cdots} & {\bmone^T} \\
    \end{bmatrix} \begin{bmatrix}
        w_1 \\ w_2 \\ \vdots \\ w_m
    \end{bmatrix} = \begin{bmatrix}
        \,\,\bmq\,\,\,\, \\ 1 \\ 1 \\ \vdots \\ 1
    \end{bmatrix}, \quad
    \textrm{or in short } A \mathbf{w} = \mathbf{b}
,\end{equation}
    where $I$ is the $d$-dimensional identity matrix and $\bmone$ is the vector of $d$ ones. From Farkas' Lemma it follows that Equation~\eqref{eqn:farkas constraint} holds if and only if
\begin{equation}~\label{eqn:farkas condition}
    \forall\, \mathbf{r} \in \RR^d,\forall\, \mathbf{t} \in \RR^{d'}, A^T \big[\begin{smallmatrix} \mathbf{r} \\ -\mathbf{t} \end{smallmatrix}\big]\geq 0 \implies \mathbf{b}^T\big[\begin{smallmatrix} \mathbf{r} \\ -\mathbf{t} \end{smallmatrix}\big] \geq 0,
\end{equation}
    which is simplified as 
\begin{equation}
    \forall\, \mathbf{r} \in \RR^d,\forall\,\mathbf{t} \in \RR^{d'}, \text{each } t_j \leq \min_{1 \leq k \leq d} p_j r_k \implies \sum_{j = 1}^{d'} t_j \leq \mathbf{r}^T \bmq,
\end{equation}
when substituting $A^T$ and $\mathbf{b}$. Since $\mathbf{t}$ is free, by pushing each $t_j$ towards its upper bound, above equation is equivalent to 
\begin{equation}~\label{eqn:p q ineq}
    \forall\, \mathbf{r} \in \RR^d,\sum_{j = 1}^{d'} \min_{1 \leq k \leq d} p_j r_k \leq \mathbf{r}^T \bmq
.\end{equation}
    Considering the left hand equals to $p_+ r_{\min} + (1 - p_+) r_{\max}$ only depending on $p_+, r_{\min}$ and $r_{\max}$, where $p_+, 1 - p_+$ are the sum of the positive and negative components of $\bmp$, respectively, we could rewrite such equation as 
\begin{equation} 
    p_+ r_{\min} + (1 - p_+) r_{\max} \leq \mathbf{r}^T \bmq, \quad 
    \forall\,r_{\min}\le r_{\max}, \forall\, \mathbf{r} \in\left[r_{\min},r_{\max}\right]^d
,\end{equation}
    or
\begin{equation} 
    p_+ r_{\min} + (1 - p_+) r_{\max} \leq \min_{\mathbf{r} \in\left[r_{\min},r_{\max}\right]^d}\mathbf{r}^T \bmq= q_+ r_{\min}  + (1 - q_+) r_{\max}, \quad
    \forall\,r_{\min}\le r_{\max}
,\end{equation}
which is 
\begin{equation}
    \left(r_{\max} - r_{\min}\right) (p_+ - q_+) \geq 0, \quad
    \forall\,r_{\min}\le r_{\max}
,\end{equation}
and is equivalent to
\begin{equation}
    p_+ \geq q_+,\text{ also }\norm{W_\rho}_1\geq\norm{W_\sigma}_1,
\end{equation}
where $q_+, 1 - q_+$ are the sum of the positive and negative components of $\bmq$, respectively.
Recall that $\cM(\rho) = \log \norm{W_\rho}_1$, $\cM(\sigma) = \log \norm{W_\sigma}_1$, and Equation~\eqref{eqn:farkas constraint} gives the existence of $\cN$. One can therefore conclude that $\rho$ can be transformed to $\sigma$ under $\PWPq$ operations if and only if $\cM(\rho) \geq \cM(\sigma)$ .
\end{proof}
\vspace{1em}

We also provide an alternative proof of Theorem~\ref{thm:magic transfer} with matrix manipulations only. We use the same notation in Equation~\eqref{eqn:sto notation}, and further assume that $\bmp = \left(\begin{smallmatrix}
    \bmp^+ \\ -\bmp^-
\end{smallmatrix}\right)$ and $\bmq = \left(\begin{smallmatrix}
    \bmq^+ \\ -\bmq^-
\end{smallmatrix}\right)$ for $\bmp^\pm, \bmq^+ \geq 0, \bmq^- > 0$. Such assumption is without loss of generality since permutation transformations are doubly stochastic (i.e., column stochastic and row stochastic).

\vspace{1em}
\begin{proof} 
($\implies$) $W \bmp = \bmq$ implies 
\begin{equation} 
    W \begin{pmatrix} \bmp^+ \\ {\bm 0} \end{pmatrix} \geq \begin{pmatrix} \bmq^+ \\ -\bmq^- \end{pmatrix}, 
    \quad W \begin{pmatrix} {\bm 0} \\ \bmp^- \end{pmatrix} \geq \begin{pmatrix} -\bmq^+ \\ \bmq^- \end{pmatrix}
.\end{equation}
Since $W, \bmp^+, \bmp^-$ have non-negative entries,
$W \begin{pmatrix} \bmp^+ \\ {\bm 0} \end{pmatrix}, W \begin{pmatrix} {\bm 0} \\ \bmp^- \end{pmatrix}  \geq 0$ and hence
\begin{equation}
    \norm{\bmp^+}_1 
    = {\bm 1}^T \begin{pmatrix} \bmp^+ \\ {\bm 0} \end{pmatrix} 
    = {\bm 1}^T W \begin{pmatrix} \bmp^+ \\ {\bm 0} \end{pmatrix} 
    \geq {\bm 1}^T \begin{pmatrix} \bmq^+ \\ {\bm 0} \end{pmatrix}  
    = \norm{\bmq^+}_1
.\end{equation}
Similar idea holds for proving $\norm{\bmp^-}_1 \geq \norm{\bmq^-}_1$. Then the result follows as $\bmp$ and $\bmq$ sum to one.

($\impliedby$) Denote $\bmq^- \in \RR^{l}$. To start with, suppose $l \geq 2$. Then denote $\hat\bmq^-$ as the vector of the first $(l - 1)$ entries of $\bmq^-$ so that
\begin{equation}
    \bmq^- = \begin{pmatrix} \hat\bmq^- \\ {\bm 1}^T\left( \bmq^- - \hat\bmq^- \right) \end{pmatrix},
    \bmq = \begin{pmatrix} \bmq^+ \\ -\hat\bmq^- \\ 1 - {\bm 1}^T\left( \bmq^+ - \hat\bmq^- \right) \end{pmatrix}
.\end{equation}
    Since $\norm{ \bmp }_1 \geq \norm{ \bmq }_1$ and $\norm{\bmp^+}_1 - \norm{\bmp^-}_1 = \norm{\bmq^+}_1 - \norm{\bmq^-}_1$, $\norm{\bmp^\pm}_1 \geq \norm{\bmq^\pm}_1$.
    The desired $W$ can then be constructed as follows:
\begin{align}
    W = \begin{pmatrix}
        W^+ & 0 \\
        0 & W^- \\
        {{\bm c}^+}^T & {{\bm c}^-}^T
    \end{pmatrix}
    \textrm{ with }
    W^+ &= \frac{\bmq^+ {\bm 1}^T}{\norm{\bmp^+}_1}, 
    W^- = \frac{\hat\bmq^- {\bm 1}^T}{\norm{\bmp^-}_1}, \textrm{ and} \\
    {\bm c}^+ &= \left( 1 - \frac{\norm{\bmq^+}_1}{\norm{\bmp^+}_1} \right) {\bm 1}, 
    {\bm c}^- = \left( 1 - \frac{\norm{\hat\bmq^-}_1}{\norm{\bmp^-}_1} \right) {\bm 1}
.\end{align}
    From above construction, observe that $A$ is column stochastic and
\begin{align}
    {{\bm c}^+}^T \bmp^+ - {{\bm c}^-}^T \bmp^- &= \left( \norm{\bmp^+}_1 - \norm{\bmq^+}_1 \right) - \left( \norm{\bmp^-}_1 - \norm{\hat\bmq^-}_1 \right) \\
    &= 1 - \left( \norm{\bmq^+}_1 - \norm{\hat\bmq^-}_1 \right) = 1 - {\bm 1}^T\left( \bmq^+ - \hat\bmq^- \right) 
.\end{align}
    Consequently, we have
\begin{align}
    W\bmp &= \begin{pmatrix}
        W^+ & 0 \\
        0 & W^- \\
        {{\bm c}^+}^T & {{\bm c}^-}^T
    \end{pmatrix} \begin{pmatrix}
        \bmp^+ \\ -\bmp^-
    \end{pmatrix}
    = \begin{pmatrix}
        W^+ \bmp^+ \\ - W^- \bmp^- \\ {c^+}^T \bmp^+ - {c^-}^T \bmp^-
    \end{pmatrix}
    = \begin{pmatrix}
        \bmq^+ \\ -\hat\bmq^- \\ 1 - {\bm 1}^T \left(\bmq^+ - \hat\bmq^-\right)
    \end{pmatrix} = \bmq
,\end{align}
    as required. In the case when $l = 0$, both $\bmp$ and $\bmq$ have no negative components and hence $W = W^+$; when $l = 1$ i.e., $\bmq$ has only one negative entry, one do not need to consider $W^-$ and hence $W = \left(\begin{smallmatrix}
        W^+ & 0 \\ {{\bm c}^+}^T & {{\bm c}^-}^T
    \end{smallmatrix}\right)$.
\end{proof}

To achieve universal quantum computation and unlock potential quantum advantages, it is essential to incorporate and manipulate a large number of identically and independently distributed copies of magic states. Consequently, understanding the transformation rate associated with converting between different magic states becomes crucial. Building on Theorem~\ref{thm:magic transfer}, we here provide the asymptotic exact magic transformation rate between magic states and establish the reversibility of these transformations under $\PWPq$ operations in the asymptotic regime.
\begin{shaded}
    \begin{theorem}\label{thm:reversible}
    For any two odd dimensional states $\rho$ and $\sigma$, the asymptotic exact magic transformation rate is given by
\begin{align}
    R_{\PWPq}(\rho\to\sigma)=\frac{\cM(\sigma)}{\cM(\rho)},
\end{align}
    which implies the reversibility of asymptotic exact magic state manipulation under $\PWPq$ operations, i.e.,
\begin{equation}
    R_{\Omega}(\rho \to \sigma) \times R_{\Omega}(\sigma \to \rho) = 1.
\end{equation}
\end{theorem}
\end{shaded}

Theorem~\ref{thm:reversible} can be intuitively understood if we consider a scenario where $m$ copies of a quantum state $\rho$ and $n$ copies of another state $\sigma$ are interchangeable under $\PWPq$ operations. In this case they own equal amounts of magic mana by Theorem~\ref{thm:magic transfer}, and hence the precise exchange rate between $\rho$ and $\sigma$ is determined by the ratio $m / n = \cM(\rho) / \cM(\sigma)$.
From this standpoint, Theorem~\ref{thm:reversible} establishes that the resource cost of preparing a magic state is balanced by the resource obtain from consuming it, with the exchange rate directly linked to the magic mana of the states involved.

\section{Physical Implementability of $\PWPq$ operations}~\label{sec:physical_imple}
Theorem~\ref{thm:magic transfer} posits that under $\PWPq$ operations, quantum states with identical magic mana can undergo reversible transformations amongst themselves. For instance, consider the Strange state $\ket{\mathbb{S}} \coloneqq (\ket{1} - \ket{2})/\sqrt{2}$ and the Norrell state $\ket{\mathbb{N}} \coloneqq (-\ket{0} + 2\ket{1} - \ket{2})/\sqrt{6}$, which have the same magic mana. According to Theorem~\ref{thm:magic transfer}, these states can be transformed reversibly into one another. However, the theorem does not specify which state possesses more ``magic'', nor does it address the comparative difficulty of transforming between these states. This leads to the pivotal question: how can we quantify the difficulty of transferring between quantum states when they exhibit the same magic mana? Addressing this question necessitates a more precise metric to establish a hierarchy of magic among states with identical magic mana.

To answer this question, we introduce a metric that quantifies the difficulty associated with transferring magic states under $\PWPq$ operations. According to the golden rule of quantum resource theory that each resource theory comes with a set of free states $\cF$ and a set of free operations $\cO$, and the minimum requirement for free operations is that they should not create resourceful states out of free states~\cite{RMP_Chitambar_2019} and in the similar spirit of \cite{jiang2021physical}, we define this quantity as follows,
\begin{definition}~\label{def:phy implement}
    For two quantum states $\rho, \sigma$, the \emph{physical implementability of magic manipulations} from $\rho$ to $\sigma$ is,
\begin{equation}
    \nu(\rho \to \sigma) = \log \min \{ \gamma \,\mid\, \cN = \sum\nolimits_j c_j \cN_j, \sum\nolimits_j c_j = 1, W_{\cN}, W_{\cN_j} \geq 0, \cN_j \in \CPTP \}
.\end{equation}
\end{definition}

By allowing the magic manipulation among these states contain transformation error, such that $\norm{\cN(\rho) - \sigma} \leq \eps$ for some error tolerance $\eps$, the physical implementability of $\varepsilon$-error state convertibility can be determined via the following SDP.
\begin{mini!}|s|
    { \cJ_{\cN_{1}}, \cJ_{\cN_{2}} }{ 2c - 1 }{}{}
    \addConstraint{ \cJ_{\cN} = \cJ_{\cN_1} - \cJ_{\cN_2}, \cJ_{\cN_{j}} \succeq 0 , \,\,\forall\, j  \text{ (HP)}}
    \addConstraint{ \tr_2 \cJ_{\cN_1} = c I, \tr_2 \cJ_{\cN_2} = (c - 1) I \text{ (TP)}}
    \addConstraint{ W_{\cN_1}, W_{\cN_2}, W_{\cN} \geq 0 \text{ (PWP)}}
    \addConstraint{ -\eps I \preceq \tr_B[(\rho^T \otimes I)\cdot\cJ_\cN] - \sigma \preceq \eps I,  \text{ (transformation within error)}}
\end{mini!}
The details of this SDP can be deferred to Appendix~\ref{appendix:sdp simplified} and the output of this SDP evaluates to $2^{\nu}$.  It is easy to see that when $\nu \geq 0$, there exists a $\PWPq$ operations to transform states and $\nu = 0$ if and only if the operations are CPTP-PWP. To achieve an estimation error within an error $\epsilon$ with a probability no less than $1-\delta$, the number of total sampling times $S$ can be estimated by Hoeffding's inequality as $S \geq 2(2c-1)^2 \log(2/\delta)/\epsilon^2$. Therefore, $\nu$ directly quantifies the hardness of quantum state convertibility statistically.

To illustrate the hardness of implementing $\PWPq$ operations physically, we consider the transformations between pairs of $\ket{\mathbb{S}}$, $\ket{\NN}$ and other two well-known magic states,
\begin{align}
    \ket{\mathbb{T}} &\coloneqq (e^{2\pi i / 9}\ket{0} + \ket{1} + e^{-2\pi i / 9}\ket{2})/\sqrt{3}, \\
    \ket{\mathbb{H}} & \textrm{ is the eigenstate of } \operatorname{QFT}_3 \textrm{ w.r.t. eigenvalue} +1
,\end{align}
where $\operatorname{QFT}_3$ is the quantum Fourier transform matrix of dimension 3.

\begin{figure}[htbp]
    \centering
    \includegraphics[width=\linewidth]{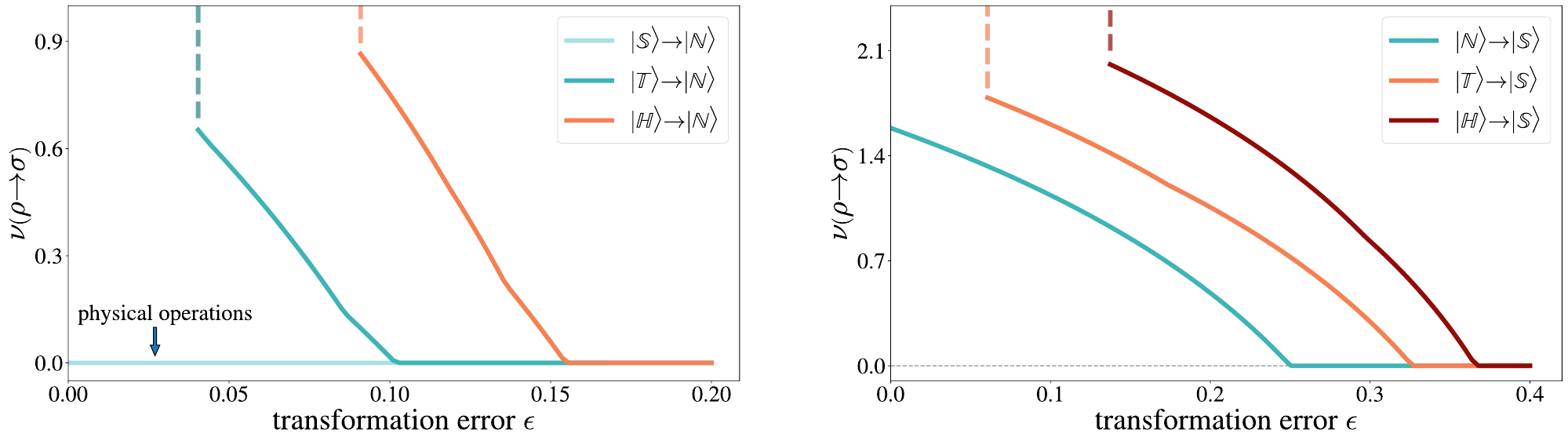}
    \caption{The physical implementability of $\PWPq$ transformations between pairs of states. Here $\eps$ is the transformation error, and the dashed line indicates the infeasibility of performing the transformations.}~\label{fig:transform w error}
\end{figure}

Figure~\ref{fig:transform w error} demonstrates that state transformations between $\ket{\mathbb{S}}$ and $\ket{\NN}$ via $\PWPq$ operations are feasible. The process is entirely physical when converting from $\ket{\mathbb{S}}$ to $\ket{\NN}$ but necessitates extra physical resources when reversing the direction. Conversely, states $\ket{\mathbb{T}}$ and $\ket{\mathbb{H}}$ not only require additional resources but also incur transformation errors in their conversion to $\ket{\NN}$ or $\ket{\mathbb{S}}$. 
Also, since the magic mana of the two states can be quantitatively ordered as $\cM(\ket{\mathbb{S}}) = \cM(\ket{\NN}) > \cM(\ket{\mathbb{T}}) > \cM(\ket{\mathbb{H}})$, these empirical observations corroborate the results detailed in Theorem~\ref{thm:magic transfer}. Despite $\ket{\NN}$ and $\ket{\mathbb{S}}$ have the same magic mana, the reversibility between $\ket{\NN}$ and $\ket{\mathbb{S}}$ is still physically hard to implement as suggested by $\nu(\ket{\NN} \to \ket{\mathbb{S}})$.

Overall, physical implementability suggests that the hardness of transforming these four states can be ranked as $\ket{\mathbb{S}}, \ket{\NN}, \ket{\mathbb{T}}, \ket{\mathbb{H}}$. This result can be extended to a finite number of resource states by initially comparing the magic mana of these states. Then for each pair of states with identical mana, we compare their physical implementability. This approach can heuristically suggest the desired hierarchy among the states. Furthermore, our experimental results indicate the transformation of certain magic states requires non-physical operations, implying the reversibility of magic state manipulation to be stemmed from the non-physicality of the operations involved.

We also note that our definition of the physical implementability of $\PWPq$ operations deviates from those in~\cite{Yuan2024}, where the HPTP map is allowed to generate resources. Specifically, their optimization program focuses on identifying the minimal sampling overhead and resource requirements necessary for quantum state convertibility. However, we impose a restriction requiring that the map be resource non-generating in our study. This adheres to the golden rules of quantum resource theory and aids in determining the additional costs associated with converting quantum states under $\PWPq$ operations, even when the two quantum states possess identical magic mana.

%%%%%%%%%%%%%%%%%%%%%%%%%%%%%%%%
\section{Concluding Remarks}

Our work demonstrates the reversibility of exact magic state manipulation under $\PWPq$ operations. We introduce the stochastic formalism of magic state transformations, implying that these relaxed operations are essentially stochastic transitions among quasi-probability distributions. This finding allows us to identify the setting of magic transformations as one that is completely governed by a unique resource measure - magic mana, thus mirroring the parallels between the resource theories of magic states and thermodynamics. 
We further propose the physical implementability of magic transformation to analyze the physical cost of guaranteeing such reversibility. This quantity can offer a deeper characterization of the hardness associated with maintaining reversibility between magic states even with the identical magic mana.

This reversible theory of magic manipulation provides new guidance on finding the smallest subset of quantum operations that can guarantee the reversibility of asymptotic magic manipulation. It also contributes to the broader comprehension of the fundamental problem of general quantum resource transformation.
% , as detailed in Table~\ref{tab:paradigm comp}
Note that the work~\cite{regula2024reversibility} demonstrates that the reversibility of general quantum resources can be recovered by expanding the set of operations into probabilistic protocols with non-vanishing probability of success. Together with this work, our findings shed light on studying the quantum resource manipulation between reversibility, success probability and positivity of allowed transformations of quantum resources. The physical implementability for reversible magic transformations also provides a new angle for understanding and exploring the quantification and manipulation of quantum resources.

\section*{Acknowledgement}
The authors are listed in alphabetical order.
This work was partially supported by the National Key R\&D Program of China (No. 2024YFE0102500), the Guangdong Provincial Quantum Science Strategic Initiative (No. GDZX2303007), the Start-up Fund (No. G0101000151) from The Hong Kong University of Science and Technology (Guangzhou), and the Education Bureau of Guangzhou Municipality.

%%%%%%%%%%%%%%%%%%%%%%%%%%%%%%%%%%%%%%%%%%%%%%%%%%%%%%%%%%%%%%%%%%%%%%%%%
% Bibliography
\bibliographystyle{alpha}
\bibliography{references}

\newcommand{\etalchar}[1]{$^{#1}$}
\begin{thebibliography}{WHSK20}

\bibitem[ADGS18]{Ahmadi2017}
Mehdi Ahmadi, Hoan~Bui Dang, Gilad Gour, and Barry~C. Sanders.
\newblock {Quantification and manipulation of magic states}.
\newblock {\em Physical Review A}, 97(6):062332, jun 2018.

\bibitem[BBC{\etalchar{+}}19]{Bravyi2018}
Sergey Bravyi, Dan Browne, Padraic Calpin, Earl Campbell, David Gosset, and
  Mark Howard.
\newblock {Simulation of quantum circuits by low-rank stabilizer
  decompositions}.
\newblock {\em Quantum}, 3:181, sep 2019.

\bibitem[BBG{\etalchar{+}}23]{Berta2023}
Mario Berta, Fernando G. S.~L. Brandão, Gilad Gour, Ludovico Lami, Martin~B.
  Plenio, Bartosz Regula, and Marco Tomamichel.
\newblock On a gap in the proof of the generalised quantum stein's lemma and
  its consequences for the reversibility of quantum resources.
\newblock {\em Quantum}, 7:1103, September 2023.

\bibitem[BCHK20]{Beverland2019}
Michael Beverland, Earl Campbell, Mark Howard, and Vadym Kliuchnikov.
\newblock {Lower bounds on the non-Clifford resources for quantum
  computations}.
\newblock {\em Quantum Science and Technology}, 5(3):035009, jun 2020.

\bibitem[BK05]{Bravyi2005}
Sergey Bravyi and Alexei Kitaev.
\newblock Universal quantum computation with ideal clifford gates and noisy
  ancillas.
\newblock {\em Physical Review A}, 71(2), February 2005.

\bibitem[Car79]{Carnot1979}
Sadi Carnot.
\newblock {\em {R{\'{e}}flexions sur la puissance motrice du feu}}, volume~26.
\newblock Vrin, 1979.

\bibitem[CG19]{RMP_Chitambar_2019}
Eric Chitambar and Gilad Gour.
\newblock Quantum resource theories.
\newblock {\em Reviews of Modern Physics}, 91(2), apr 2019.

\bibitem[GJB{\etalchar{+}}18]{gour2018quantum}
Gilad Gour, David Jennings, Francesco Buscemi, Runyao Duan, and Iman Marvian.
\newblock Quantum majorization and a complete set of entropic conditions for
  quantum thermodynamics.
\newblock {\em Nature communications}, 9(1):5352, 2018.

\bibitem[Got97]{gottesman1997stabilizer}
Daniel Gottesman.
\newblock {\em Stabilizer codes and quantum error correction}.
\newblock California Institute of Technology, 1997.

\bibitem[Gro06a]{Gross_2006a}
D.~Gross.
\newblock Hudson's theorem for finite-dimensional quantum systems.
\newblock {\em Journal of Mathematical Physics}, 47(12), dec 2006.

\bibitem[Gro06b]{Gross_2006b}
D.~Gross.
\newblock Non-negative wigner functions in prime dimensions.
\newblock {\em Applied Physics B}, 86(3):367--370, dec 2006.

\bibitem[HC17]{Howard2016}
Mark Howard and Earl Campbell.
\newblock {Application of a Resource Theory for Magic States to Fault-Tolerant
  Quantum Computing}.
\newblock {\em Physical Review Letters}, 118(9):090501, mar 2017.

\bibitem[HHH98]{Horodecki_1998}
Micha\l{} Horodecki, Pawe\l{} Horodecki, and Ryszard Horodecki.
\newblock Mixed-state entanglement and distillation: Is there a ``bound''
  entanglement in nature?
\newblock {\em Phys. Rev. Lett.}, 80:5239--5242, Jun 1998.

\bibitem[JW23]{Jiang2021}
Jiaqing Jiang and Xin Wang.
\newblock {Lower Bound for the T Count Via Unitary Stabilizer Nullity}.
\newblock {\em Physical Review Applied}, 19(3):034052, mar 2023.

\bibitem[JWW21]{jiang2021physical}
Jiaqing Jiang, Kun Wang, and Xin Wang.
\newblock Physical implementability of linear maps and its application in error
  mitigation.
\newblock {\em Quantum}, 5:600, December 2021.

\bibitem[KJ22]{koukoulekidis2022constraints}
Nikolaos Koukoulekidis and David Jennings.
\newblock Constraints on magic state protocols from the statistical mechanics
  of wigner negativity.
\newblock {\em npj Quantum Information}, 8(1):42, April 2022.

\bibitem[Kni05]{knill2005quantum}
Emanuel Knill.
\newblock Quantum computing with realistically noisy devices.
\newblock {\em Nature}, 434(7029):39--44, 2005.

\bibitem[LBT19]{liu2019one}
Zi-Wen Liu, Kaifeng Bu, and Ryuji Takagi.
\newblock One-shot operational quantum resource theory.
\newblock {\em Physical review letters}, 123(2):020401, 2019.

\bibitem[LOH22]{Leone2022}
Lorenzo Leone, Salvatore F~E Oliviero, and Alioscia Hamma.
\newblock {Stabilizer r{\'{e}}nyi entropy}.
\newblock {\em Physical Review Letters}, 128(5):50402, 2022.

\bibitem[ME12]{Mari2012}
A.~Mari and J.~Eisert.
\newblock {Positive wigner functions render classical simulation of quantum
  computation efficient}.
\newblock {\em Physical Review Letters}, 109(23):1--7, 2012.

\bibitem[Per96]{Peres1996}
Asher Peres.
\newblock {Separability Criterion for Density Matrices}.
\newblock {\em Physical Review Letters}, 77(8):1413--1415, apr 1996.

\bibitem[RBTL20]{regula2020benchmarking}
Bartosz Regula, Kaifeng Bu, Ryuji Takagi, and Zi-Wen Liu.
\newblock Benchmarking one-shot distillation in general quantum resource
  theories.
\newblock {\em Physical Review A}, 101(6):062315, 2020.

\bibitem[RL24]{regula2024reversibility}
Bartosz Regula and Ludovico Lami.
\newblock Reversibility of quantum resources through probabilistic protocols.
\newblock {\em Nature Communications}, 15(1):3096, April 2024.

\bibitem[SC19]{Seddon_2019}
James~R. Seddon and Earl~T. Campbell.
\newblock Quantifying magic for multi-qubit operations.
\newblock {\em Proceedings of the Royal Society A: Mathematical, Physical and
  Engineering Sciences}, 475(2227):20190251, jul 2019.

\bibitem[Sho97]{Shor_1997}
Peter~W. Shor.
\newblock Polynomial-time algorithms for prime factorization and discrete
  logarithms on a quantum computer.
\newblock {\em {SIAM} Journal on Computing}, 26(5):1484--1509, oct 1997.

\bibitem[TR19]{Takagi_2019}
Ryuji Takagi and Bartosz Regula.
\newblock General resource theories in quantum mechanics and beyond:
  Operational characterization via discrimination tasks.
\newblock {\em Physical Review X}, 9(3), sep 2019.

\bibitem[VC01]{Vidal2001}
Guifr{\'{e}} Vidal and J~I Cirac.
\newblock {Irreversibility in Asymptotic Manipulations of Entanglement}.
\newblock {\em Physical Review Letters}, 86(25):5803--5806, jun 2001.

\bibitem[VFGE12]{Veitch2012}
Victor Veitch, Christopher Ferrie, David Gross, and Joseph Emerson.
\newblock {Negative quasi-probability as a resource for quantum computation}.
\newblock {\em New Journal of Physics}, 14:1--15, 2012.

\bibitem[VHGE14]{Veitch2014}
Victor Veitch, S.~A. {Hamed Mousavian}, Daniel Gottesman, and Joseph Emerson.
\newblock {The resource theory of stabilizer quantum computation}.
\newblock {\em New Journal of Physics}, 16:1--23, 2014.

\bibitem[WCZZ23]{wang2023reversible}
Xin Wang, Yu-Ao Chen, Lei Zhang, and Chenghong Zhu.
\newblock Reversible entanglement beyond quantum operations, 2023.

\bibitem[WHSK20]{Wang_2020}
Yuchen Wang, Zixuan Hu, Barry~C. Sanders, and Sabre Kais.
\newblock Qudits and high-dimensional quantum computing.
\newblock {\em Frontiers in Physics}, 8, nov 2020.

\bibitem[Woo87]{WOOTTERS19871}
William~K Wootters.
\newblock A wigner-function formulation of finite-state quantum mechanics.
\newblock {\em Annals of Physics}, 176(1):1--21, 1987.

\bibitem[WWS18]{Wang2018}
Xin Wang, Mark~M Wilde, and Yuan Su.
\newblock {Efficiently computable bounds for magic state distillation}.
\newblock {\em Physical Review Letters}, 124(9):090505, dec 2018.

\bibitem[WWS19]{WWS19}
Xin Wang, Mark~M Wilde, and Yuan Su.
\newblock {Quantifying the magic of quantum channels}.
\newblock {\em New Journal of Physics}, 21(10):103002, oct 2019.

\bibitem[WY16]{Winter2016}
Andreas Winter and Dong Yang.
\newblock Operational resource theory of coherence.
\newblock {\em Physical Review Letters}, 116(12), March 2016.

\bibitem[YRTG24]{Yuan2024}
Xiao Yuan, Bartosz Regula, Ryuji Takagi, and Mile Gu.
\newblock Virtual quantum resource distillation.
\newblock {\em Physical Review Letters}, 132(5), February 2024.

\bibitem[ZWC24]{Zhao2024}
Xuanqiang Zhao, Xin Wang, and Giulio Chiribella.
\newblock {Shadow simulation of quantum processes}.
\newblock {\em arXiv preprint arXiv:2401.14934}, jan 2024.

\end{thebibliography}

%%%%%%%%% SUPPLEMENTAL MATERIAL %%%%%%%%%

\newpage

\appendix
\setcounter{subsection}{0}
\setcounter{table}{0}
\setcounter{figure}{0}

\vspace{3cm}
% \onecolumngrid
% \vspace{2cm}

\begin{center}
\Large{\textbf{Supplementary Material} \\ \textbf{
}}
\end{center}

\renewcommand{\theequation}{S\arabic{equation}}
% \numberwithin{equation}{section}
\renewcommand{\thesubsection}{\normalsize{Supplementary Note \arabic{subsection}}}
\renewcommand{\theproposition}{S\arabic{proposition}}
\renewcommand{\thedefinition}{S\arabic{definition}}
\renewcommand{\thefigure}{S\arabic{figure}}
\setcounter{equation}{0}
\setcounter{table}{0}
\setcounter{section}{0}
\setcounter{proposition}{0}
\setcounter{definition}{0}
\setcounter{figure}{0}

% \tableofcontents
% 

%%%%%%%%%%%%%%%%%%%%%%%%%%%%%%%%%%%%%%%%%%%%%%%%%%%%%%%%%%%%%%%%%%%%%%%%%%%
\section{The discrete Wigner function}\label{appendix:Wigner function}
Let $\cH_d$ denotes a Hilbert space of dimension $d$, $\{ \ket{j} \}_{j = 0, \cdots, d-1}$ denote the standard computational basis. Let $\cL(\cH_d)$ be the space of operators mapping $\cH_d$ to itself. For odd dimension $d$, the unitary boost and shift operators $X, Z \in \cL(\cH_d)$ are defined as~\cite{WWS19}: 
\begin{equation}
X\ket{j} = \ket{j \oplus 1},\quad Z\ket{j} = w^j \ket{j},
\end{equation}
where $w = e^{2\pi i/d}$ and $\oplus$ denotes addition modulo $d$. The \textit{discrete phase space} of a single $d$-level system is $\mathbb{Z}_d \times \mathbb{Z}_d$, which can be associated with a $d\times d$ cubic lattice. For a given point in the discrete phase space $\bmu=\left(u_1, u_2\right) \in \mathbb{Z}_d \times \mathbb{Z}_d$, the Heisenberg-Weyl operators are given by \begin{equation}
    T_\bmu=\tau^{-u_1 u_2} Z^{u_1} X^{u_2},
\end{equation}
where $\tau=e^{(d+1) \pi i / d}$. These operators form a group, the Heisenberg-Weyl group, and are the main ingredient for representing quantum systems in finite phase space. The case of non-prime dimension can be understood to be a tensor product of $T_\bmu$ with odd dimension. For each point $\bmu \in \mathbb{Z}_d \times \mathbb{Z}_d$ in the discrete phase space, there is a phase-space point operator $A_\bmu$ defined as 
\begin{equation}
A_0\coloneqq\frac{1}{d} \sum_\bmu T_\bmu,\quad A_\bmu \coloneqq T_\bmu A_0 T_\bmu^{\dagger}. 
\end{equation}
The discrete Wigner function of a state $\rho$ at the point  $\bmu$ is then defined as 
\begin{equation}
    \cW_\rho(\bmu)\coloneqq\frac{1}{d} \tr\left[A_\bmu \rho\right].
\end{equation}
More generally, we can replace $\rho$ with $H$ for the Wigner function of a Hermitian operator $X$.
There are several useful properties of the set $\{A_\bmu\}_\bmu$ as follows:
\begin{enumerate}
    \item $A_\bmu$ is Hermitian;
    \item $\sum_u A_\bmu/d = \mathbb{I}$;
    \item $\trace{A_\bmu A_{\bmu'}} = d \delta_{\bmu \bmu'}$;
    \item $\trace{A_\bmu} = 1$;
    \item $X = \sum_\bmu W_{X}(\bmu)A_\bmu$;
    \item $\{ A_\bmu\}_\bmu = \{ {A_\bmu}^T\}_\bmu$.  
\end{enumerate}

We call a Hermitian operator $X$ has positive discrete Wigner functions (PWFs) if $\cW_X(\bmu) \geq 0, \forall \bmu \in \mathbb{Z}_d \times \mathbb{Z}_d$. 

However, operations with PWFs including stabilizer operations cannot yield negative discrete Wigner functions, making negative quasi-probability a vital resource for quantum speedup in stabilizer computation~\cite{Veitch2012}. 

\section{Computation of Physical Implementability}~\label{appendix:sdp simplified}

From Definition~\ref{def:phy implement}, the physical implementability pf $\PWPq$ transformations from $\rho$ to $\sigma$ can be determined by the following SDP:
\begin{mini!}|s|
    { \cN_j }{ \sum_j \abs{c_j} }{}{}
    \addConstraint{ \cN = \sum_j \cN_j, \cJ_{\cN_j} \succeq 0 \,\,\forall\, j  \text{ (HP)}}
    \addConstraint{ \tr_2 J_{\cN_j} = c_j I \,\,\forall\, j, \sum_j c_j = 1 \text{ (TP)}}
    \addConstraint{ W_{\cN_j} \geq 0 \,\,\forall\, j  \text{ (PWP)}}
    \addConstraint{ \cN(\rho) = \sigma,  \text{ (Exact transformation)}}
\end{mini!}
\noindent Observe that for any $\cN_j, \cN_k$ that satisfy above constraints,
\begin{equation}
    \tr_2 J_{\cN_j + \cN_k} = \left(c_j + c_k\right) I, 
    \quad W_{\cN_j + \cN_k} \geq 0,
    \quad \cN_j + \cN_k \succeq 0
.\end{equation}
Then one can denote $\cN_1 = \sum_{j: c_j \geq 0} \cN_j$ and $\cN_2 = \sum_{j: c_j \geq 0} \cN_j$, and the optimization problem is equivalent to
\begin{mini!}|s|
    { \cJ_{\cN_{1}}, \cJ_{\cN_{2}} }{ 2c - 1 }{}{}
    \addConstraint{ \cJ_{\cN} = \cJ_{\cN_1} - \cJ_{\cN_2}, \cJ_{\cN_{j}} \succeq 0 \,\,\forall\, j  \text{ (HP)}}
    \addConstraint{ \tr_2 J_{\cN_1} = c I, \tr_2 J_{\cN_2} = (c - 1) I \text{ (TP)}}
    \addConstraint{ W_{\cN_1}, W_{\cN_2} \geq 0 \text{ (PWP)}}
    \addConstraint{ \tr_B[(\rho^T \otimes I)\cdot\cJ_\cN] = \sigma,  \text{ (Exact transformation)}}
\end{mini!}
Finally, when the transformation error is allowed, such that $\norm{\cN(\rho) - \sigma}_\infty < \eps$, the optimization problem can be reduced to
\begin{mini!}|s|
    { \cJ_{\cN_{1}}, \cJ_{\cN_{2}} }{ 2c - 1 }{}{}
    \addConstraint{ \cJ_{\cN} = \cJ_{\cN_1} - \cJ_{\cN_2}, \cJ_{\cN_{j}} \succeq 0 \,\,\forall\, j  \text{ (HP)}}
    \addConstraint{ \tr_2 \cJ_{\cN_1} = c I, \tr_2 \cJ_{\cN_2} = (c - 1) I \text{ (TP)}}
    \addConstraint{ W_{\cN_1}, W_{\cN_2}, W_{\cN} \geq 0 \text{ (PWP)}}
    \addConstraint{ -\eps I \preceq \tr_B[(\rho^T \otimes I)\cdot\cJ_\cN] - \sigma \preceq \eps I,  \text{ (transformation within error)}}
\end{mini!}
%%%%%%%%%%%%%%%%%%%%%%%%%%%%%%%%%%%%%%%%%%%%%%%%%%%%%%%%%%%%%%%%%%%%%%%%%%%

\end{document}